\documentclass[twocolumn,reprint,secnumarabic,aps,prd,nofootinbib]{revtex4-2}
\usepackage[utf8]{inputenc}
\usepackage[english]{babel}
\usepackage{amsmath}
\usepackage{mathrsfs}
\usepackage{amsfonts}
\usepackage{amssymb}
\usepackage{mathtools}
\usepackage{amsthm}
\usepackage{bbm}
\usepackage{color}
\usepackage{footmisc}
\usepackage[hyperindex,breaklinks]{hyperref}
\hypersetup{colorlinks   = true, 
	urlcolor     = black, 
	linkcolor    = blue, 
	citecolor   = teal}

\usepackage{cleveref}
\Crefname{equation}{Eq.}{Eqs.}
\Crefname{figure}{Fig.}{Figs.}
\Crefname{tabular}{Tab.}{Tabs.}

\usepackage{graphicx}
\usepackage{enumerate}
\usepackage{cases}
\usepackage{graphicx}
\usepackage{tikz-cd} 
\usepackage{ytableau}

\theoremstyle{plain}

\newtheorem{thm}{Theorem}[section] 
\newtheorem{lem}[thm]{Lemma}
\newtheorem{prop}[thm]{Proposition}

\newtheorem*{defi*}{Definition}

\newtheorem*{thm*}{Theorem}

\theoremstyle{definition}

\newcommand{\nn}{\mathbb{N}}
\newcommand{\zz}{\mathbb{Z}}
\newcommand{\rr}{\mathbb{R}}
\newcommand{\cc}{\mathbb{C}}
\newcommand{\Abb}{\mathbb{A}}

\newcommand{\dd}{\mathrm{\normalfont d}}

\newcommand{\GL}{\mathrm{GL}}
\newcommand{\Oo}{\mathrm{O}}
\newcommand{\Pp}{\mathrm{P}}
\newcommand{\PO}{\mathrm{PO}}

\newcommand{\Ad}{\mathrm{Ad}}
\newcommand{\ad}{\mathrm{ad}}

\newcommand{\SL}{\mathrm{SL}}

\newcommand{\SU}{\mathrm{SU}}
\newcommand{\Uu}{\mathrm{U}}

\newcommand{\pfrak}{\mathfrak{p}}

\newcommand{\sfrak}{\mathfrak{s}}

\newcommand{\ufrak}{\mathfrak{u}}
\newcommand{\ofrak}{\mathfrak{o}}

\newcommand{\pgl}{\mathfrak{pgl}}

\newcommand{\yy}{\mathbf{y}}

\newcommand{\CP}{\mathbb{CP}}
\newcommand{\RP}{\mathbb{RP}}
\newcommand{\PGL}{\mathrm{PGL}}
\newcommand{\Xx}{\mathsf{X}}

\newcommand{\Gg}{\mathsf{G}}
\newcommand{\GSM}{G_\mathrm{SM}}

\newcommand{\Hcal}{\mathcal{H}}

\newcommand{\Ocal}{\mathcal{O}}
\newcommand{\Pcal}{\mathcal{P}}

\newcommand{\Afrak}{\mathfrak{A}}

\newcommand{\Extalt}{\mathchoice{{\textstyle\bigwedge}}%
    {{\bigwedge}}%
    {{\textstyle\wedge}}%
    {{\scriptstyle\wedge}}}

\newcommand{\glfrak}{\mathfrak{gl}}

\newcommand{\Ss}{\mathrm{S}}

\newcommand{\totimes}{\,\tilde{\otimes}\,}

\tikzset{
    labl/.style={anchor=north, rotate=90, inner sep=1.mm}
}

\begin{document}
\title{Similarities between projective quantum fields and the Standard Model}
\author{Daniel Spitz}
\email{daniel.spitz@mis.mpg.de}
\affiliation{Institut f\"ur theoretische Physik, Ruprecht-Karls-Universit\"at Heidelberg, \mbox{Philosophenweg 16}, 69120 Heidelberg, Germany}
\affiliation{Max-Planck-Institut f\"ur Mathematik in den Naturwissenschaften, \mbox{Inselstraße 22, 04103 Leipzig}, Germany}


\begin{abstract}
Many homogeneous, four-dimensional space-time geometries can be considered within real projective geometry, which yields a mathematically well-defined framework for their deformations and limits without the appearance of singularities. 
Focussing on generalized unitary transformation behavior, projective quantum fields can be axiomatically introduced, which transform smoothly under geometry deformations and limits.
Connections on the related projective frame bundles provide gauge fields with gauge group $\PGL_5\rr$.
For Poincaré geometry, on operator level only $\Pp(\GL_2\rr\times \GL_3\rr)\cong \rr_{\neq 0}\times \PGL_2\rr\times \PGL_3\rr$ gauge bosons can interact non-trivially with other projective quantum fields from the non- to ultra-relativistic limits.
The corresponding propagating, complexified gauge bosons come with the Standard Model gauge group $G_{\mathrm{SM}}=(\Uu(1)\times \SU(2)\times \SU(3))/\zz_6$.
Physical scale transformations can act as global gauge transformations and their spontaneous breaking can lead to masses for the projective quantum fields including the $\SU(2)$ gauge bosons.
Projective quantum fields, which are irreducible both with respect to the Lie algebra $\pgl_5\rr$ and the Poincaré group, form Dirac fermions and $\GSM$ gauge bosons interact with them similar to the Standard Model.
For homogeneous, curved Lorentz geometries a gauge group similar to the gauge group of metric-affine gravity appears.
\end{abstract}

\maketitle


\section{Introduction}
The Standard Model of particle physics (SM)~\cite{Donoghue:1992dd, Schwartz:2014sze,Langacker:2017uah} has seen a tremendous amount of experimental validation. 
It incorporates a gauge symmetry with gauge group $\GSM = (\Uu(1)\times \SU(2)\times \SU(3))/\zz_6$, the cyclic $\zz_6$ subgroup generated by $(\exp(\pi i /3),-1_{2\times 2},\exp(2\pi i /3)\cdot 1_{3\times 3})$~\cite{Baez:2009dj}.
Three generations of Dirac fermions tensored with particular $\GSM$ representations interact with the gauge bosons.
Oscillations among the generations occur, and the fermions and $W^\pm, Z$ gauge bosons acquire effective masses via the Higgs mechanism.

A unified framework to understand this peculiar structure, at least in parts and optimally based on widely accepted physical concepts, is lacking.
Concerning the gauge symmetry, the exterior algebra $\Extalt\cc^5$ provided an early, remarkably simple finding. 
With the non-trivial identification $\Ss(\Uu(2)\times\Uu(3))\cong \GSM$, it yields the gauge group representation of the SM fermions augmented by two trivial representations~\cite{Georgi:1974sy,Baez:2009dj}.
This stimulated far-reaching research on grand unified theories extending the SM gauge group~\cite{Langacker:1980js}, but evidence for the related proton decays remains absent~\cite{Super-Kamiokande:2014otb}.
Not aiming for grand unification, a number of mathematical constructions can give rise to $\GSM$ gauge theory and further SM structures.
They are based on e.g.~non-commutative geometry~\cite{Connes:1990qp,Chamseddine:1996zu}, the octonions~\cite{Dubois-Violette:2016kzx, Furey:2018drh, Todorov:2018mwd, Krasnov:2019auj} or Twistor theory~\cite{Woit:2021bmb}.
This paper describes similarities between so-called projective quantum fields and the SM, based on projective models of space-time geometries, generalized unitarity and irreducibility assumptions.
Moreover, indications for the presence of gravity appear.

The central mathematical observation for this is that many homogeneous space-time geometries can be considered within real projective geometry, which is similar to Weyl geometry, but not necessarily metric-based.
Then, space-times and their symmetry groups are defined only up to non-zero prefactors, which allows for the description of space-time geometry limits via limits of projective matrix products~\cite{cooper2018limits}.
Examples for such limits are the flattening of de Sitter geometry, which is curved, or the non-relativistic (Galilei) limit of Poincaré geometry.
Quantum field theories (QFTs) can be formulated in this setting via generalized unitary transformation behavior under space-time symmetry transformations and projectivization of the field operators.
Such projective quantum fields have recently been introduced by the author~\cite{Spitz:2024similarities} and shown to transform smoothly under (homogeneous) geometry deformations and limits, similarly for their correlation functions.
Geometry limits can reduce the space-time dimension for projective quantum fields.

Gauge fields with gauge group $\PGL_5\rr$ appear in this context as connections of the projective variant of the frame bundle and can implement the invariance under local changes of projective reference frames.
This is similar to connections in metric-affine gravity~\cite{Hehl:1994ue} and generalizes gauged affine symmetries in the projective setting, including gauged dilatations.

The generalized unitary transformation behavior of projective quantum fields links $\PGL_5\rr$ gauge transformations to space-time symmetry transformations, but preserving locality.
This way, the presence of gauge interactions in the non- and ultra-relativistic limits of Poincaré geometry requires that on operator level the gauge group is basically reduced to $\Pp(\GL_2\rr\times \GL_3\rr)\cong \rr_{\neq 0}\times \PGL_2\rr \times \PGL_3\rr$, as we prove in this work.
For homogeneous, curved Lorentz geometries, the reduced gauge group $\Pp(\rr_{\neq 0}\times \GL_4\rr)\cong \GL_4\rr$ together with an Abelian subgroup of $\PGL_5\rr$ isomorphic to $\rr^4$ appears, if gauge interactions, which survive the flattening process to Poincaré geometry, are considered.

Physical scale transformations can act as global gauge transformations on projective quantum fields restricted to Poincaré geometry.
As for the Stueckelberg mechanism~\cite{Stueckelberg:1938zz, Ghilencea:2019jux}, this can lead to masses for those fields, which are not invariant under such transformations, if the corresponding global gauge symmetry is spontaneously broken.
In particular, this applies to the $\PGL_2\rr$ gauge bosons.
Masses for other projective quantum fields can be effectively generated as in the SM.

It has been shown in~\cite{Spitz:2024similarities} that fermionic projective quantum fields on Poincaré geometry, which transform irreducibly with respect to the Lie algebra $\pgl_5\rr$ of $\PGL_5\rr$ and the Poincaré group, come with uneven exterior powers of the fundamental representation $\CP^4_{\pgl_5\rr}$ of $\pgl_5\rr$ and their dual representations.
Taking their complexifications into account, Poincaré transformations act on them as for Dirac fermions and gauge transformations in the reduced gauge group act similar to how $\GSM$ acts on one generation of the SM fermions.
The SM gauge group $\Ss(\Uu(2)\times \Uu(3))\cong\GSM$ appears, if the propagating, complexified gauge bosons are considered, i.e., the gauge bosons for the compact part of the complexified Lie algebra of $\Pp(\GL_2\rr\times\GL_3\rr)$.
The involved gauge algebra map can be physically motivated based on representation equivalence and is similar to how spinor representations appear for QFTs on Poincaré geometry.
Unlike the SM, in our model gauge and Poincaré transformations act not via tensor product representations on the fermions.

This manuscript is organized as follows.
In \Cref{SecProjQuantumFields} we review projective quantum fields, which incorporate the generalized unitary transformation behavior.
\Cref{SecPGL5RGaugeTheory} briefly describes how $\PGL_5\rr$ gauge fields can appear in this context.
In \Cref{SecGaugeGroupReduction} we prove the gauge group reduction theorem central to this work.
\Cref{SecScaleTrafos} considers the action of physical scale transformations as a 1-parameter gauge subgroup.
Finally, in \Cref{SecSimilaritiesSM} we discuss how this leads to similarities with the SM.
We conclude in \Cref{SecOutlook} with an outlook, and a few appendices provide proofs for statements in the main text.

\section{Quantum fields on projective geometries}\label{SecProjQuantumFields}
This work deals with quantum fields on homogeneous space-time geometries $(\Xx,\Oo)$, which consist of a space-time $\Xx$ (the model space) and a Lie group $\Oo$ acting transitively on $\Xx$ (the structure group).
The mathematical theory of such Klein geometries is described e.g.~in~\cite{cooper2018limits}, which we summarize here.
We focus on geometries, for which $\Xx$ is a subspace of the four-dimensional real projective space $\RP^4$ and $\Oo$ is a subgroup of $\PGL_5\rr$, denoted $(\Xx,\Oo)<(\RP^4,\PGL_5\rr)$.
The projective setting with its equivalence up to coordinate rescalings, which is similar to Weyl geometry with its invariance under metric rescalings, allows us later to consistently describe geometry deformations and limits without the appearance of singularities.
For introductions to projective geometry we refer to~\cite{ovsienko2004projective,Gallier2011}.
Important examples of such geometries have the structure groups
\begin{align}\label{EqOMatrices}
&\PO((p_0,q_0),\dots,(p_k,q_k))\nonumber\\
&\qquad:=\Pp\left(\begin{matrix}
\Oo(p_0,q_0) & 0 & \cdots & 0 \\
\rr^{m_1\times m_0} & \Oo(p_1,q_1) &  & 0\\
\vdots & \vdots & \ddots & \vdots \\
\rr^{m_k\times m_0} & \rr^{m_k\times m_1} & \cdots & \Oo(p_k,q_k)
\end{matrix}\right)\,,
\end{align}
where $m_i:=p_i+q_i$, $\sum_i m_i = 5$ and which act transitively on
\begin{align}
&\Xx((p_0,q_0),\dots,(p_k,q_k))\nonumber\\
& := \{[x_0,\dots,x_{m-1}]\in \RP^4\,|\, -x_0^2-\ldots-x_{p_0-1}^2\nonumber\\
&\qquad\qquad\qquad\qquad +x_{p_0}^2+\ldots+x_{m_0-1}^2 < 0\}\,.
\end{align}
We write 
\begin{align}
&\Gg((p_0,q_0),\dots,(p_k,q_k))\nonumber\\
&:= (\Xx((p_0,q_0),\dots,(p_k,q_k)),\Pp\Oo((p_0,q_0),\dots, (p_k,q_k)))\,.
\end{align}
For instance, $\Gg(4,1)$ is the projective model of four-dimensional de Sitter geometry and $\Gg(3,2)$ is the projective model of four-dimensional anti-de Sitter geometry, while $\Gg((1),(3,1))$ with
\begin{equation}\label{EqPoincareModelSpaceA31}
\Xx((1),(3,1)) = \{[x_0,\dots,x_4]\,|\, x_0\neq 0\} = \Abb^{3,1}
\end{equation}
is the projective model of Poincaré geometry.
This can be identified with Poincaré geometry itself, the Poincaré group isomorphically represented by projective $5\times 5$ matrices.
In \Cref{EqPoincareModelSpaceA31}, $\Abb^{3,1}$ denotes the four-dimensional affine space endowed with the Minkowski metric tensor $(\eta_{\mu\nu}) = \mathrm{diag}(-1,-1,-1,+1)$ and the identification with $\Xx((1),(3,1))$ is via 
\begin{equation}\label{EqDiffeoPsi}
\xi:[x_0,\dots,x_4] \mapsto (x_1/x_0,\dots,x_4/x_0)\,.
\end{equation}
Note that for notational consistency with regard to the metric signatures, the last component on the right-hand side of \eqref{EqDiffeoPsi} denotes the time coordinate.
When referring to specific space-time geometries such as Poincaré geometry, in this work mostly their projective models are meant.

A space-time geometry $\Gg=(\Xx,\Oo)<(\RP^4,\PGL_5\rr)$ can be deformed by $[g]\in \PGL_5\rr$ acting on the model space by projective matrix multiplication, $\Xx\ni [x]\mapsto [g\cdot x]$, and conjugating the elements of $\Oo$, $\Oo\ni [h]\mapsto \Ad_{[g]}[h] := [g\cdot h \cdot g^{-1}]$.
The pair $[g]_*\Gg:=([g]\cdot \Xx,\Ad_{[g]}\Oo)$ is again a geometry.

Replacing $[g]$ by a sequence $[b_n]\in \PGL_5\rr$, the $n\to\infty$ limit of $[b_n]_*\Gg$ can be defined as follows~\cite{cooper2018limits}.
A Lie group $\Oo'$ is called the conjugacy limit of $\Oo$ via $[b_n]$, if every $[h']\in \Oo'$ is the limit of some sequence $[h_n]\in \Ad_{[b_n]}\Oo$ and if every accumulation point of every sequence $[h_n]\in \Ad_{[b_n]}\Oo$ lies in $\Oo'$.
This way, the Lie algebra $\ofrak'$ of $\Oo'$ is given by $\ofrak'=\lim_{n\to\infty} \Ad_{[b_n]}\ofrak$~\cite{trettel2019families}, where $\ofrak$ is the Lie algebra of $\Oo$ and convergence is defined as for conjugacy limits.%
\footnote{Actually, it has been shown in~\cite{Spitz:2024similarities} that conjugacy limits act on Lie algebra level isomorphically to compositions of Lie algebra contractions.} 
Finally, the sequence of geometries $[b_n]_*\Gg < (\RP^4,\PGL_5\rr)$ converges to the geometry $\Gg'=(\Xx',\Oo')$, if $\Oo'$ is the conjugacy limit of $\Oo$ via $[b_n]$ and there exists $[z]\in \Xx'\subset \RP^4$, such that for all $n$ sufficiently large $[z]\in [b_n]\cdot \Xx$. 
In fact, all such limits of $\Gg((p_0,q_0),\dots,(p_k,q_k))< (\RP^4,\PGL_5\rr)$ have been classified, see~\cite{cooper2018limits}.

As an example, consider Poincaré geometry $\Gg((1),(3,1))$ and a Lorentz-boost in the 3-direction.
The sequence $[b_n]=\Pp(\mathrm{diag}(\exp(-n),1,1,1,\exp(-n)))$ acts on the lower-right $2\times 2$ submatrix of the boost generator as
\begin{align}\label{EqBoostLimit}
\Pp\left(\begin{matrix}
0 & i \\
i & 0
\end{matrix}\right)&\mapsto \Pp\left(\begin{matrix}
0 & i e^n\\
i e^{-n} & 0
\end{matrix}\right) \nonumber\\
&\qquad\quad= \Pp\left(\begin{matrix}
0 & i \\
i e^{-2n} & 0
\end{matrix}\right) \to \Pp\left(\begin{matrix}
0 & i\\
0 & 0
\end{matrix}\right)
\end{align}
for $n\to\infty$.
In the projective setting, the multiplication by $\exp(-n)$ is an identity map and provides the basis for well-defined limit matrices.
Due to \eqref{EqBoostLimit}, the limit process turns Lorentz boosts into Galilean velocity additions, corresponding to the infinite speed of light limit.
The conjugacy limit of $\PO((1),(3,1))$ via $[b_n]$ is isomorphic to the projective Galilei group:
\begin{equation}
\lim_{n\to\infty} \Ad_{[b_n]} \PO((1),(3,1)) = \Ad_{\tau}\PO((1),(1),(3))\,,
\end{equation}
where $\tau = (0)\,(1\, 2\, 3\, 4)$ permutes coordinates.
The model space remains unaltered:
\begin{equation}
[b_n]\cdot \Xx((1),(3,1)) = \Xx((1),(3,1)) = \Xx((1),(1),(3))\,, 
\end{equation}
such that the limit geometry is the projective model of Galilei geometry:
\begin{equation}
[b_n]_*\Gg((1),(3,1))\to \tau_* \Gg((1),(1),(3))\,.
\end{equation}
For further examples of geometry limits we refer to~\cite{Spitz:2024similarities}.

Quantum fields depend on the geometry of space-time through projective unitary representations of the structure group, which act on the field operators~\cite{Weinberg:1995mt}.
In order to make use of the described framework for geometry limits, the author axiomatically introduced projective quantum fields in~\cite{Spitz:2024similarities}.
For the purposes of the present work, only their representation-theoretic behavior is of interest.%
\footnote{Accordingly, here we do not consider the boundedness of field operators upon smearing with test functions and the restriction of their domain to certain dense subspaces of the Hilbert space.} 
Let $\overline{\PGL_5\rr}$ denote the (universal) double cover of $\PGL_5\rr$.
Then, a projective quantum field is a tuple $\hat{\Ocal}=(U,\rho,\{[\hat{\Ocal}([x])]\,|\, [x]\in\RP^4\})$ consisting of a projective unitary $\PGL_5\rr$ representation $U$ on a Hilbert space $\Hcal$, a finite-dimensional, complex $\overline{\PGL_5\rr}$ representation $\rho$ and an equivalence class of operators $[\hat{\Ocal}([x])]$ for every $[x]\in\RP^4$, defined modulo global $C^\infty(\RP^4,\rr_{\neq 0})$ prefactors, such that for all $[x]\in \RP^4$ and all $\alpha=1,\dots,\dim\rho$ their representatives obey generalized unitary transformation behavior for arbitrary $[g]\in\PGL_5\rr$:
\begin{equation}\label{EqGlobalCovariance}
U([g]) \hat{\Ocal}_\alpha([x])U^\dagger([g]) = \sum_{\beta=1}^{\dim\rho} \rho_{\alpha\beta}([g^{-1}]) \hat{\Ocal}_\beta([g\cdot x])\,.
\end{equation}
On a geometry $(\Xx,\Oo)<(\RP^4,\PGL_5\rr)$, the projective quantum field is given by restriction:
\begin{equation}\label{EqGeometryRestriction}
\hat{\Ocal}|_{(\Xx,\Oo)}:=(U|_{\Oo},\rho,\{[\hat{\Ocal}([x])]\,|\, [x]\in\Xx\})\,.
\end{equation}

In this work we only consider projective quantum fields such as $\hat{\Ocal}$, which do not possess internal degrees of freedom.
By means of the geometry restriction \eqref{EqGeometryRestriction}, \Cref{EqGlobalCovariance} indeed generalizes the unitary transformation behavior of quantum fields on fixed homogeneous geometries such as de Sitter, anti-de Sitter or Poincaré geometry. 
It has been shown in~\cite{Spitz:2024similarities} that such projective quantum fields transform smoothly under geometry deformations and limits without singularities appearing, but their space-time support can shift to space-time boundaries and other lower-dimensional space-time subspaces.
A corresponding QFT can therefore be dimensionally reduced through geometry limits.

We define an algebra $\Afrak_{\hat{\Ocal}}(\Xx)$ of projective correlators for the projective quantum field $\hat{\Ocal}$ as
\begin{align}\label{EqCorrelatorAlgebra}
\Afrak_{\hat{\Ocal}}(\Xx):=&\big\{ [\hat{\Ocal}(f_1,\Xx)^{(\dagger)}\otimes \ldots \otimes \hat{\Ocal}(f_{\ell},\Xx)^{(\dagger)}]\nonumber\\
&\qquad\qquad\qquad \,\big| f_i\in C^\infty(\RP^4), \ell\in\nn\big\}\,,
\end{align}
where smeared field operator representatives are given by
\begin{equation}
\hat{\Ocal}(f,\Xx):=\int_\Xx \dd^4[x]\, f([x]) \hat{\Ocal}([x])
\end{equation}
and formally encode renormalization of the projective correlators~\cite{Haag1996Local}.
The tensor products in \Cref{EqCorrelatorAlgebra} are defined only with respect to the indices corresponding to the action of $\rho$, so that the domain of the projective correlators remains the original Hilbert space.
The superscript $(\dagger)$ indicates taking the adjoint of individual field operators or not.
The author has shown in~\cite{Spitz:2024similarities} that expectation values of all projective correlators in $\Afrak_{\hat{\Ocal}}(\Xx)$ remain well-defined under geometry deformations and limits, but their space-time support can shift to space-time boundaries and other lower-dimensional space-time subspaces, analogously to projective quantum fields.
In particular, this applies to their infrared and ultraviolet limits, for which the space-time dimension reduces to 3.

\section{$\PGL_5\rr$ gauge theory}\label{SecPGL5RGaugeTheory}
In the framework of projective quantum fields, $\PGL_5\rr$ gauge fields appear as connections of the projective frame bundle, which is based on projective (reference) frames.
We first introduce projective vector fields and projective frames, which provides a summary of the corresponding considerations in~\cite{Spitz:2024similarities}, and subsequently sketch the occurrence of $\PGL_5\rr$ gauge fields.

A projective vector field on a four-dimensional model space $\Xx$ is an equivalence class $[W]$ of vector fields on $\Xx$, where $W\sim\tilde{W}$, if there exists $\lambda\in C^\infty(\Xx,\rr_{\neq 0})$: $\tilde{W} = \lambda W$ everywhere on $\Xx$.
A projective frame on $\Xx$ is defined from a non-projective frame $(e_1,\dots,e_4)$ on $\Xx$ as $([e_1],\ldots,[e_4],[e_5]=[e_1+\ldots+e_4])$.
Projective frames on $\Xx$ consist of five projective vector fields, so that the projective frame is uniquely determined up to a common $C^\infty(\Xx,\rr_{\neq 0})$ prefactor.
Indeed, given two projective frames $([e_1], \ldots,[e_5])$, $([f_1],\ldots,[f_5])$, if $[e_B]=[f_B]$ holds for all $B=1,\dots,5$, then $e_B=\lambda f_B$ for all $B$ and a single $\lambda\in C^\infty(\Xx,\rr_{\neq 0})$.
In terms of the projective frame $([e_B])$, the projective vector field $[W]$ can be non-uniquely decomposed as $[W] = [\sum_B W_B e_B]$.
What is uniquely fixed by $[W]$ is the set of homogeneous coordinates of $[W]$ with respect to $([e_B])$:
\begin{align}\label{EqHomogeneousCoords}
&\bigg\{(\lambda W_1',\dots,\lambda W_4')\,\bigg|\, W_\mu'\in C^\infty(\Xx),\lambda\in C^\infty(\Xx,\rr_{\neq 0}),\nonumber\\
&\qquad\qquad\qquad\qquad\qquad\qquad \bigg[\sum_{\nu =1}^4 W_\nu' e_\nu\bigg] = [W]\bigg\}\,.
\end{align}

We denote the set of projective frames on $\Xx$ by $\Pcal(\Xx)$.
Locally, elements $[G]\in C^\infty(\Xx,\PGL_5\rr)$ act transitively on a projective frame $([e_B])\in \Pcal(\Xx)$ via right multiplication:
\begin{equation}
([e_1],\dots,[e_5])\mapsto ([e_1],\dots,[e_5])\cdot [G]\,.
\end{equation}
This yields the projective frame bundle over $\Xx$, i.e., the principal $\PGL_5\rr$-bundle
\begin{equation}\label{EqProjectiveFrameBundle}
\PGL_5\rr\to \Pcal(\Xx)\to \Xx\,.
\end{equation}

A connection 1-form of this bundle uniquely fixes a gauge 1-form $A([x])$, which takes values in the Lie algebra $\pgl_5\rr$.
Projectivization yields the projective gauge 1-form $[A([x])]$,%
\footnote{Projective 1-forms on $\Xx$ are defined analogously to projective vector fields, i.e., they are equivalence classes $[\alpha]$ of 1-forms on $\Xx$, where $\alpha\sim \tilde{\alpha}$, if there exists $\lambda\in C^\infty(\Xx,\rr_{\neq 0})$: $\tilde{\alpha} = \lambda \alpha$ everywhere on $\Xx$.
Considering the projectivization of usual connections is similar to projective connections, for which we expect equivalent results.
For clarity we here focus on the former construction.} 
which non-uniquely decomposes for a given projective frame $([e_B])$ or rather its dual into the components $[A_B([x])]$.
We note that the gauge group $\PGL_5\rr$ contains the gauge group $\GL_4\rr$ of metric-affine gravity by means of an embedding similar to \Cref{EqOMatrices}, see also the end of \Cref{SecGaugeGroupReduction}.

Based on this, a projective $\PGL_5\rr$ gauge quantum field $\hat{A}=(U,\rho_A,\{[\hat{A}([x])]\,|\,[x]\in\RP^4\})$ on the Hilbert space $\Hcal$ is a projective quantum field, which locally transforms under a gauge transformation $[V]\in C^\infty(\RP^4,\PGL_5\rr)$ as
\begin{align}\label{EqGaugeTrafoGaugeField}
[\hat{A}_B([x])]\mapsto & [V([x])\cdot \hat{A}_B([x])\cdot V^{-1}([x])\nonumber\\
&\qquad\qquad\qquad - i V([x]) \cdot  \partial_B V^{-1}([x])]
\end{align}
for all $[x]\in\RP^4$.
Upon restriction to a geometry $(\Xx,\Oo)<(\RP^4,\PGL_5\rr)$, gauge transformations and their action \eqref{EqGaugeTrafoGaugeField} on $\hat{A}$ are restricted to $\Xx$.

Deformations and limits of geometries act on projective $\PGL_5\rr$ gauge quantum fields as for projective quantum fields.
The gauge transformation $[V]$ canonically acts on the field operators of the projective quantum field $\hat{\Ocal}$, which has no internal degrees of freedom, as
\begin{equation}\label{EqActionGaugeTrafoQuantumField}
[\hat{\Ocal}([x])]\mapsto [\rho([V([x])]) \hat{\Ocal}([x])]\,.
\end{equation}
A covariant derivative acting on the $[\hat{\Ocal}([x])]$ can be defined similarly to spin connections~\cite{Ortin:2015hya, Floerchinger:2021uyo}:
\begin{equation}\label{EqCovariantDerivativeMatter}
\hat{\nabla}_B([x]) = [\partial_B - i \tilde{\rho}([\hat{A}_B([x])])]\,,
\end{equation}
where $\tilde{\rho}$ denotes the Lie algebra representation corresponding to $\rho$.

\section{Gauge group reduction}\label{SecGaugeGroupReduction}
We turn to implications of certain geometry limits for gauge-matter interactions on Poincaré and homogeneous, curved Lorentz geometries.
Specifically, we show that on operator level only a subgroup conjugate to $\Pp(\GL_2\rr\times \GL_3\rr)$ within the full gauge group $\PGL_5\rr$ can act non-trivially on the field operators from the non- to the ultra-relativistic limits of Poincaré geometry, and only gauge bosons for the gauge subgroup $\Pp(\GL_4\rr\times \rr_{\neq 0})$ together with an Abelian $\PGL_5\rr$ subgroup isomorphic to $\rr^4$ survive e.g.~in the flat de Sitter limit.

First, we consider Poincaré geometry $\Gg((1),(3,1))$ and geometry deformations by
\begin{equation}\label{EqctauNonrelUltrarelInterpolation}
[c(\kappa)] = \Pp\left(\begin{matrix}
\kappa & & \\
 & 1_{3\times 3} &\\
& & \kappa
\end{matrix}\right)\quad \text{for} \quad \kappa > 0\,.
\end{equation}
Around \Cref{EqBoostLimit}, we have shown that $\lim_{\kappa\to 0} [c(\kappa)]_* \Gg((1),(3,1))$ is the (non-relativistic) Galilei geometry $\Gg((1),(1),(3))$ up to the permutation $\tau$.
It can be shown that in the $\kappa \to \infty$ limit, the geometry becomes the (ultra-relativistic) Carroll geometry $\Gg((1),(3),(1))$.
Therefore, the action of $[c(\kappa)]$ changes the coordinate speed of light and can be used as a geometry deformation to interpolate between the non- and ultra-relativistic limits of Poincaré geometry.

We say that the projective gauge quantum field $\hat{\Ocal}$ and the projective $\PGL_5\rr$ gauge quantum field $\hat{A}$ interact non-trivially in the non- and ultra-relativistic limits of Poincaré geometry, if for all $[x]\in\Xx((1),(3,1))$ both the $\kappa\to 0$ limit and the $\kappa\to \infty$ limit of the operator
\begin{equation}\label{EqGaugeIAOperatorckappadeform}
 \Ad_{U([c(\kappa)])}[\hat{\nabla}_B([x])\hat{\Ocal}([x])]_{B} 
\end{equation}
are non-trivial as operators acting on the Hilbert space.
Physically, this assumption can be motivated by analogy with gauge interactions in the SM.
Concerning the non-triviality of electromagnetic gauge-matter interactions in the non-relativistic limit, we note that many cross-sections for processes in quantum electrodynamics (QED) remain non-trivial in this limit, see e.g.~\cite{Peskin:1995ev}.
Non-relativistic quantum chromodynamics has been successfully employed in the description of a range of related particle production phenomena~\cite{Bodwin:1994jh, Cho:1995vh, Cho:1995ce, Butenschoen:2009zy}, and the $W^\pm,Z$ bosons are naturally described non-relativistically at energies below the electroweak symmetry breaking scale due to their comparably large mass.
The SM is also non-trivial in its Carroll limits: QED and other QFTs admit a description on Carroll geometry, see e.g.~\cite{Basu:2018dub, Bagchi:2019xfx, Banerjee:2020qjj} and note applications of this in the framework of hydrodynamics~\cite{Ciambelli:2018wre, Bagchi:2023rwd, Bagchi:2023ysc}.

We can now state the gauge group reduction theorem, which is central to this work.
Intuitively, considering limits of interaction operators among the gauge and matter fields deformed by $[c(\kappa)]$ probes the presence of gauge interactions on operator level in the non- and ultra-relativistic limits of Poincaré geometry.

\begin{thm}\label{ThmGaugeGroupReduction}
Let $\hat{\Ocal}$ be a projective quantum field, $\hat{A}$ a projective $\PGL_5\rr$ gauge quantum field and $\PGL_5\rr$ gauge transformations act on $[\hat{\Ocal}([x])]$ as in \Cref{EqActionGaugeTrafoQuantumField}. 
Assume $\hat{\Ocal}$ and $\hat{A}$ interact non-trivially in the non- and ultra-relativistic limits of Poincaré geometry.
Then maximally gauge bosons for the gauge subgroup $\Ad_\tau\Pp(\GL_2\rr\times \GL_3\rr)< \PGL_5\rr$ can act non-trivially on the projective correlators in $\Afrak_{\hat{\Ocal}}(\Xx((1),(3,1)))$ via covariant derivatives.
\end{thm}

\begin{proof}
The generalized unitary transformation behavior \eqref{EqGlobalCovariance} yields for any $[x]\in\Xx((1),(3,1))$, $\alpha = 1,\dots,\dim\rho$ and any $B=1,\dots, 5$:
\begin{align}
&\sum_{C,\beta} c(\kappa)_{BC} \Ad_{U([c(\kappa)])}\tilde{\rho}_{\alpha\beta}([\hat{A}_C([x])])\hat{\Ocal}_\beta([x])\nonumber\\
& =\;  \sum_{\beta,\gamma}\tilde{\rho}_{\alpha\beta}([\hat{A}_B([c(\kappa)x])])\rho_{\beta\gamma}([c(\kappa)^{-1}])\hat{\Ocal}_\gamma([c(\kappa)x])\nonumber\\
& =\; \sum_{\beta,\gamma}\rho_{\alpha\beta}([c(\kappa)^{-1}])\nonumber\\
&\qquad\;\times \tilde{\rho}_{\beta\gamma}([c(\kappa)\hat{A}_B([c(\kappa)x])c(\kappa)^{-1}])\hat{\Ocal}_\gamma([c(\kappa)x])\,,
\end{align}
where representatives of the projective field operators are considered.
The non-triviality assumption and surjectivity of the map $\Xx((1),(3,1))\ni[x]\mapsto [c(\kappa)x]\in \Xx((1),(3,1))$ for all $\kappa >  0$ imply that upon their action on the field operators of $\hat{\Ocal}$ via the representation $\tilde{\rho}$ all conjugated gauge field operators
\begin{equation}
[c(\kappa) \hat{A}_B([x]) c(\kappa)^{-1}]
\end{equation}
must remain non-zero for both the $\kappa \to 0$ limit and the $\kappa\to\infty$ limit.
Therefore, maximally gauge field operators for the gauge subgroup $\Ad_{\tau}\Pp(\GL_2\rr\times\GL_3\rr)<\PGL_5\rr$ can interact non-trivially with the field operators of $\hat{\Ocal}$.
Indeed, with 
\begin{equation}
\Ad_{\tau^{-1}} ([c(\kappa)]) = \Pp\left(\begin{matrix}
\kappa \cdot 1_{2\times 2} & \\
 & 1_{3\times 3}
\end{matrix}\right)
\end{equation}
and $A,B,C,D$ corresponding matrix blocks:
\begin{equation}
\Ad_{\Ad_{\tau^{-1}}([c(\kappa)])} \Pp \left(\begin{matrix}
A & B\\
C & D
\end{matrix}\right)  = \Pp\left(\begin{matrix}
A & \kappa B\\
C/\kappa & D
\end{matrix}\right)\,,
\end{equation}
for which both the $\kappa \to 0$ and the $\kappa\to\infty$ limits are non-zero, if and only if one of the following cases holds:
\begin{enumerate}[(i)]
\item $B,C=0$, or
\item $A,B,D=0$, or
\item $A,C,D=0$.
\end{enumerate}
The subalgebra of $\pgl_5\rr$, which is given by projective matrices of type (i), has maximal dimension among these.

The extension of this statement to the full projective correlator algebra $\Afrak_{\hat{\Ocal}}(\Xx((1),(3,1)))$ is straight-forward.
\end{proof}

\Cref{ThmGaugeGroupReduction} provides a statement on operator level.
We therefore expect that the reduction of the $\PGL_5\rr$ gauge group to $\Ad_{\tau}\Pp(\GL_2\rr\times \GL_3\rr)$ is preserved in the classical limit and appears in corresponding action functionals.
Further considerations regarding the latter are deferred until \Cref{SecSimilaritiesSM}.

Shedding further light on the structure of the reduced gauge group, in \Cref{AppendixGroupIso} we prove that the map $\zeta:\rr_{\neq 0}\times \PGL_2\rr\times\PGL_3\rr\to \Pp(\GL_2\rr\times\GL_3\rr)$,
\begin{equation}\label{EqZeta}
\zeta(a,[g],[h]) = \Pp\left(\begin{matrix}
a^3 g & \\
 & a^{-2} h
\end{matrix}\right)\,,
\end{equation}
is a group isomorphism, where $g$ is an arbitrary of the two unit-determinant representatives of $[g]\in\PGL_2\rr$ and $h$ is the unique unit-determinant representative of $[h]\in\PGL_3\rr$.
This is analogous to $\GSM\cong \Ss(\Uu(2)\times \Uu(3))$~\cite{Georgi:1974sy,Baez:2009dj}.

A variant of \Cref{ThmGaugeGroupReduction} can also be proven for homogeneous, curved Lorentz geometries of the form
\begin{equation}
[g]_*\Gg(4,1) = ([g]\cdot \Xx(4,1), \Ad_{[g]} \Pp\Oo(4,1))
\end{equation}
for arbitrary $[g]\in\PGL_5\rr$.
For these, on operator level only gauge bosons for the gauge subgroup $\Ad_{[g]}  \Pp(\rr_{\neq 0}\times \GL_4\rr)$ or an Abelian gauge subgroup isomorphic to $\rr^4$ interact non-trivially with the field operators of $\hat{\Ocal}$ in the Poincaré limit.
The explicit statement along with its proof is given in \Cref{AppendixGaugeGrpReductionHomogLorentz}.
With the isomorphism $\Pp(\rr_{\neq 0}\times \GL_4\rr)\cong \GL_4\rr$, part of the reduced gauge group can be identified with part of the gauge group of metric-affine gravity~\cite{Hehl:1994ue}.

Finally, we note that for Poincaré geometry the presence of gauge interactions from non- to ultra-relativistic limits can be alternatively probed with
\begin{equation}
[\tilde{c}(\kappa)] = \Pp\left(\begin{matrix}
1_{4\times 4} & \\
& \kappa
\end{matrix}\right)
\end{equation}
for $\kappa > 0$ instead of $[c(\kappa)]$.
The same arguments as before then lead to the reduced gauge group $\Pp(\GL_4\rr\times \rr_{\neq 0})\cong \GL_4\rr$, which can be identified with part of the reduced gauge group for homogeneous, curved Lorentz geometries $[g]_*\Gg(4,1)$.
In fact, \Cref{LemmaRedGaugeGroupLimits} in \Cref{AppendixGaugeGrpReductionHomogLorentz} shows that the reduced gauge group for geometries of the latter type remains isomorphic to $\Pp(\rr_{\neq 0}\times \GL_4\rr)$ also in its geometry limits.
Still, we remind ourselves that the space-time support of projective quantum fields on $[g]_*\Gg(4,1)$ can reduce dimensionally in its geometry limits such as the Poincaré limit, so that $\hat{\Ocal}$ and $\hat{A}$ restricted to Poincaré geometry can not be fully recovered from their restriction to $[g]_*\Gg(4,1)$.

Considering Poincaré geometry itself, neither $\Ad_\tau\Pp(\GL_2\rr\times \GL_3\rr)$ nor $\Pp(\GL_4\rr\times \rr_{\neq 0})$ seems to be preferable.
The author does not have a decisive argument regarding the choice of $[c(\kappa)]$ or $[\tilde{c}(\kappa)]$ and focuses primarily on the former case and the corresponding reduced gauge group for Poincaré geometry.
While this will lead to similarities with the SM, we note that the other choice would lead to further similarities with gravity.
It remains to be investigated, in how far both can be mathematically realized simultaneously in the framework of projective quantum fields without internal degrees of freedom.
To conclude this section, we also note that the choice of $[c(\kappa)]$ or $[\tilde{c}(\kappa)]$ bears similarities with taking either the electric or the magnetic Carroll limit of QED~\cite{Bagchi:2016bcd}.%
\footnote{To be more explicit, given a 4-vector field $V=(\mathbf{V},V_4)$ in the notation of our paper, the electric limit corresponds to $V\mapsto (\epsilon \mathbf{V},V_4)$ for $\epsilon\to 0$, while the magnetic limit corresponds to $V\mapsto(\mathbf{V},\epsilon V_4)$ for $\epsilon\to 0$.}
Yet, while electric and magnetic Carroll limits are typically defined on field level, the action of $[c(\kappa)]$ or $[\tilde{c}(\kappa)]$ via the map $\Ad_U$ encodes the corresponding action on the space-time geometry in the projective setting.

\section{Global physical scale transformations}\label{SecScaleTrafos}
We now consider physical scale transformations on Poincaré geometry as gauge transformations.
The identification $\Xx((1),(3,1)) = \Abb^{3,1}$ has been via the map $\xi$ given in \Cref{EqDiffeoPsi}.
A physical scale transformation acts on $[x]\in\Xx((1),(3,1))$ as $\xi([x])\mapsto e^s \xi([x])$ for $s \in \rr$, which is equivalent to
\begin{equation}\label{EqScaleTrafoPoincare}
[x]\mapsto [d(s)\cdot x] := \Pp\left(\begin{matrix}
1 & \\
& e^s \cdot 1_{4\times 4}
\end{matrix}\right)\cdot [x]\,.
\end{equation}
Hence, commutativity with $U([d(s)])$ probes physical scale invariance on operator level.

The dilatations $[d(s)]$ also form a 1-parameter subgroup of the reduced gauge group $\Ad_{\tau}\Pp(\GL_2\rr\times \GL_3\rr)$ for Poincaré geometry and can act on the field operators of $\hat{\Ocal}$ via the global gauge transformation $[\hat{\Ocal}([x])]\mapsto [\rho([d(s)]) \hat{\Ocal}([x])]$.
We show that such gauge transformations change mass eigenspace contributions to the field operator representatives.
For this we write translations in the (temporal) 4-direction of $\Xx((1),(3,1))$ on operator level as $\exp(i\hat{H}y_4)$ for $y_4\in\rr$ and translations in the (spatial) 1-, 2- and 3-directions as $\exp(-i\mathbf{\hat{P}}\yy)$ for $\yy=(y_1,y_2,y_3)\in \rr^3$ with $\mathbf{\hat{P}}\yy=\hat{P}_1 y_1+\hat{P}_2 y_2 + \hat{P}_3 y_3$.
The squared mass operator then reads $\hat{M}^2 = \hat{H}^2-\mathbf{\hat{P}}^2$.
Furthermore, we write $U([d(s)]) = \exp(i\hat{D} s)$, where the operator $i\hat{D}$ satisfies the commutation relations
\begin{equation}
[i\hat{D},i\hat{H}] = i \hat{H}\,,\qquad [i\hat{D},i\hat{P}_j] = i\hat{P}_j
\end{equation}
for $j=1,2,3$, such that
\begin{equation}
\ad_{i\hat{D}} \hat{M}^2 := [i\hat{D},\hat{M}^2] = 2 \hat{M}^2\,.
\end{equation}
Therefore,
\begin{equation}\label{EqAdUMSq}
\Ad_{U([d(s)])}\hat{M}^2 = \exp(s\, \ad_{i\hat{D}}) \hat{M}^2 = e^{2s} \hat{M}^2\,,
\end{equation}
where we employed basic Lie theory~\cite{kirillov2008introduction} for the first equality.
The operator $\hat{M}^2$ is self-adjoint, so there exists a Hilbert space basis of mass eigenstates $|m\rangle$ with $\hat{M}^2|m\rangle = m^2|m\rangle$.
We note that due to \Cref{EqAdUMSq}:
\begin{align}
\hat{M}^2 U([d(s)])|m\rangle = &\;U([d(s)])\Ad_{U([d(s)^{-1}])}\hat{M}^2|m\rangle\nonumber\\
 =&\; e^{-2s} m^2 U([d(s)])|m\rangle\,,
\end{align}
so that 
\begin{equation}\label{EqUdsMinus1MassEigenstate}
U([d(s)])|m\rangle = |e^{-s} m\rangle\,.
\end{equation}

With respect to the mass eigenspaces of the Hilbert space, the field operator representatives $\hat{\Ocal}_\alpha([x])$ can be decomposed as
\begin{equation}\label{EqMassContributionsFieldOperators}
\hat{\Ocal}_\alpha([x]) = \int \dd m\int  \dd m' |m\rangle \langle m| \hat{\Ocal}_\alpha([x])|m'\rangle \langle m'|\,.
\end{equation}
On the other hand, the gauge transformed field operator representatives $\rho([d(s)])\hat{\Ocal}([x])$ have the following decomposition:
\begin{align}
&\sum_\beta \rho_{\alpha\beta}([d(s)])\hat{\Ocal}_\beta([x])\nonumber\\
&=\;\Ad_{U([d(s)^{-1}])}\hat{\Ocal}_\alpha([d(s)\cdot x])\nonumber\\
&=\; \int  \dd m\int \dd m' |m\rangle  \langle e^{-s} m| \hat{\Ocal}_\alpha([d(s)\cdot x])|e^{-s} m'\rangle \langle m'|\,,
\end{align}
where we employed the generalized unitary transformation behavior \eqref{EqGlobalCovariance} and \Cref{EqUdsMinus1MassEigenstate}.
Indeed, the gauge-transformed field operator representatives $\rho([d(s)])\hat{\Ocal}([x])$ come with a different mass eigenspace decomposition compared to \Cref{EqMassContributionsFieldOperators}.

As a special case, consider field operators $[\hat{\Ocal}_\alpha([x])]$ with a definite (effective) mass $m\in [0,\infty)$, so that for all $[x]\in \Xx((1),(3,1))$ and all $\alpha$: $[\hat{M}^2,\hat{\Ocal}_\alpha([x])] = m^2 \hat{\Ocal}_\alpha([x])$.
\Cref{EqAdUMSq} together with the generalized unitary transformation behavior \eqref{EqGlobalCovariance} yields
\begin{align}
&\bigg[\hat{M}^2,\sum_\beta \rho_{\alpha\beta}([d(s)])\hat{\Ocal}_\beta([x])\bigg] \nonumber\\
&\qquad= \Ad_{U([d(s)^{-1}])} [\Ad_{U([d(s)])}\hat{M}^2,\hat{\Ocal}_\alpha([d(s)\cdot x])]\nonumber\\
&\qquad= \Ad_{U([d(s)^{-1}])} [e^{2s} \hat{M^2},\hat{\Ocal}_\alpha([d(s)\cdot x])]\nonumber\\
&\qquad= e^{2s}m^2 \sum_\beta \rho_{\alpha\beta}([d(s)])\hat{\Ocal}_\beta([x])\,.
\end{align}
The gauge-transformed field operator representatives $\rho([d(s)])\hat{\Ocal}([x])$ have the altered squared mass $e^{2s} m^2$.

The matrices $[d(s)]$ do not commute with a $\PGL_2\rr$ subgroup of the reduced gauge group.
Specifically, using the group isomorphism $\zeta$ of \Cref{EqZeta}, \Cref{AppendixConformalWeight} shows that
\begin{equation}\label{EqGrpCommutatorNeqOne}
[d(s)^{-1} \, k^{-1}\, d(s)\, k] \neq [1]
\end{equation}
for $s\neq 0$ and $[1]\neq [k]\in \Ad_{\tau} \zeta(1,\PGL_2\rr,[1])$, which is isomorphic to $\PGL_2\rr$.
The group commutator \eqref{EqGrpCommutatorNeqOne} equals the projective unit matrix for all ${[k]\in \Ad_{\tau}\Pp(1_{2\times 2}\times \GL_3)}$, which is isomorphic to $\rr_{\neq 0}\times \PGL_3\rr$.
Therefore, the gauge bosons of a $\PGL_2\rr$ gauge subgroup do not commute with the global gauge transformations by $[d(s)]$.

The invariance under global gauge transformations by $[d(s)]$ can be spontaneously broken by non-zero expectation values of suitable operators.
This can imply non-vanishing masses for those field operators, which are not invariant under the gauge transformations by $[d(s)]$.
In particular, this applies to the $\PGL_2\rr$ gauge bosons.
Such a phenomenon would be reminiscent of the Stueckelberg mechanism in Weyl conformal geometry~\cite{Stueckelberg:1938zz, Ghilencea:2019jux}.
Details are left for future work.

We note that also for the considered non-compact gauge group $\Ad_{\tau}\Pp(\GL_2\rr\times \GL_3\rr)$, the spontaneous breaking of the global $[d(s)]$ gauge symmetry would not be in conflict with Elitzur's theorem~\cite{Elitzur:1975im}.
Indeed, if finitely many space-time points are considered, the latter implies that only locally varying gauge transformations must leave operators with non-vanishing expectation values invariant, at least for the maximal compact subgroup of the reduced gauge group, which is isomorphic to $\PO(2)\times \PO(3)$.%
\footnote{The author is aware only of variants of Elitzur's theorem for compact gauge groups, see e.g.~\cite{Elitzur:1975im, Itzykson:1989sx}.}

\section{Similarities with the Standard Model}\label{SecSimilaritiesSM}
Based on the gauge group reduction and gauged physical scale transformations, a range of structures related to projective quantum fields reveals similarities with the SM.
Using the results of~\cite{Spitz:2024similarities}, we first describe a classification of projective quantum fields based on their irreducibility.
This leads to the appearance of Dirac fermions, if Poincaré transformations are considered.
We discuss how gauge transformations in the reduced gauge group act on these.
Subsequently, implications for particle physics models in the action formalism are derived.

Projective quantum fields can be characterized according to irreducibility of the Lie algebra representation $\tilde{\rho}$ corresponding to $\rho$.
Along with Schur module-like constructions for projective quantum fields (detailed in~\cite{Spitz:2024similarities}), this leads to the description of the fundamental building blocks of projective quantum fields without internal degrees of freedom.
This is similar to how the spin representations appear for QFTs on Poincaré geometry.
The irreducible projective quantum fields are labelled by a pair of Young diagrams $(\lambda,\lambda')$, so that $\tilde{\rho} = \CP^4_{(\lambda,\lambda')}$.
The latter denotes the irreducible representation of $\pgl_5\rr$ constructed according to the Young diagrams $(\lambda,\lambda')$ from partly symmetrized, partly anti-symmetrized and dualized tensor products of $\CP^4_{\pgl_5\rr}$, for which $\pgl_5\rr$ acts via projective matrices on $\CP^4$.%
\footnote{For tensors of second or higher order, representatives of $\pgl_5\rr$ elements act on the tensor products of $\cc^5$, and are only finally projectivized.} 
We write $\rho=\rho_{(\lambda,\lambda')}$, if $\tilde{\rho} = \CP^4_{(\lambda,\lambda')}$.
Irreducible projective quantum fields with such $\tilde{\rho}$ are denoted as $\hat{\Ocal}_{(\lambda,\lambda')}$.
The particular case of interest in the present work will be treated in more detail below.

Projective quantum fields on Poincaré geometry can also be characterized according to their irreducible behavior under Poincaré transformations.%
\footnote{These can be regarded as the fundamental building blocks of projective quantum fields on Poincaré geometry, since all projective unitary representations of the Poincaré group decompose into a direct sum or direct integral of the irreducible projective unitary representations of the Poincaré group~\cite{mautner1950unitary}.}
 In fact, all fermionic projective quantum fields, for which $\tilde{\rho}$ is irreducible as a $\pgl_5\rr$ representation and the restriction $U|_{\PO((1),(3,1))}$ is irreducible as a representation of the Poincaré group, transform under Poincaré transformations as Dirac fermions, proven in~\cite{Spitz:2024similarities}.
Only one diagram of the pair $(\lambda,\lambda')$ can be non-empty and must consist of a single column with an uneven number of boxes $\#\lambda$ or $\#\lambda'$.
Leading for $\tilde{\rho}$ only to anti-symmetric tensor representations with an uneven number of $\CP^4_{\pgl_5\rr}$ factors and their duals, this demonstrates the spin-statistics relation for projective quantum fields as composites without internal degrees of freedom, constructed from projective quantum fields with $\tilde{\rho}=\CP^4_{\pgl_5\rr}$.
Bosonic such projective quantum fields would violate the spin-statistics relation for similar reasons~\cite{Spitz:2024similarities}.

Therefore, for the full gauge group $\PGL_5\rr$, the $\pgl_5\rr$-valued gauge field $[\hat{A}([x])]$ acts on the field operators $[\hat{\Ocal}_{(\lambda,\emptyset)}([x])]$ or $[\hat{\Ocal}_{(\emptyset,\lambda)}([x])]$ via the representations 
\begin{equation}
(\Pp\tilde{\Extalt}^{\#\lambda} \cc^5)_{\pgl_5\rr}\quad\mathrm{or}\quad (\Pp\tilde{\Extalt}^{\#\lambda} \cc^5)^*_{\pgl_5\rr}
\end{equation}
for uneven $\#\lambda$, respectively.
Considered for the reduced gauge group $\Pp(\GL_2\rr\times \GL_3\rr)$,%
\footnote{The conjugation of $\Pp(\GL_2\rr\times \GL_3\rr)$ by $\tau$ is a group isomorphism and yields an equivalence of representations.} 
we have with $\pfrak(\glfrak_2\rr \oplus \glfrak_3\rr)\cong \rr \oplus\pgl_2\rr\oplus \pgl_3\rr$ (due to \Cref{PropReducedGaugeGroupIso}) the isomorphisms
\begin{subequations}
\begin{align}
&\Pp\tilde{\Extalt}^1 \cc^5|_{\pfrak(\glfrak_2\rr\oplus\glfrak_3\rr)} \cong \CP^4|_{\pfrak(\glfrak_2\rr\oplus \glfrak_3\rr)}\nonumber\\
&\qquad\cong (\cc_{+1}\totimes \CP^1 \totimes 1) \oplus (\cc_{-2/3} \totimes 1 \totimes \CP^2)\,, \label{EqLambda1IrrepReduced}\\
&\Pp\tilde{\Extalt}^3 \cc^5|_{\pfrak(\glfrak_2\rr\oplus\glfrak_3\rr)} \nonumber\\
&\qquad\cong (\cc_{-2}\totimes 1\totimes 1)\oplus (\cc_{-1/3}\totimes \CP^{1,*}\totimes \CP^{2,*})\nonumber\\
&\qquad\qquad\quad\qquad\qquad\qquad\oplus (\cc_{+4/3}\totimes 1\totimes \CP^2)\,. \label{EqLambda3IrrepReduced}
\end{align}
\end{subequations}
Here, $\cc_{\kappa/3}$ for $\kappa\in \rr$ denotes the representation $\rr\ni a\mapsto \kappa a$, where we call $\kappa/3$ the weak hypercharge, by analogy with the SM.
$\CP^{m-1}:=\CP^{m-1}_{\pgl_m\rr}$ denotes the fundamental complex representation of $\pgl_m\rr$ with dual $\CP^{m-1,*}$ and the trivial representation is denoted $1$.
The decompositions in \Cref{EqLambda1IrrepReduced,EqLambda3IrrepReduced} are into irreducible representations of $\rr\oplus \pgl_2\rr\oplus \pgl_3\rr$ and the weak hypercharges are unique up to a linear map.

These results suggest a particle physics model, which on Poincaré geometry includes a projective $\PGL_5\rr$ gauge quantum field $\hat{A}$ and the fermionic, irreducible, Poincaré-irreducible projective quantum fields $\hat{\Ocal}_{(\lambda,\emptyset)}$ and $\hat{\Ocal}_{(\emptyset,\lambda)}$ for column-only $\lambda$ with $\#\lambda$ uneven.
As described around \Cref{EqHomogeneousCoords}, on Poincaré geometry the projective gauge quantum field $[\hat{A}_B([x])]$ uniquely fixes a non-projective gauge quantum field $\hat{A}_\mu(y)$, $\mu=1,\dots,4$, $y\in \Abb^{3,1}$, up to multiplication by $C^\infty(\Abb^{3,1},\rr_{\neq 0})$ elements.
An action functional constructed from the projective $[\hat{A}_B([x])]$ is thus the same as an action functional constructed from $\hat{A}_\mu(y)$, which is furthermore invariant under local Weyl rescalings and reflections.

The SM augmented by Weyl gauge symmetry has been studied e.g.~in~\cite{Ghilencea:2018thl, Ghilencea:2021lpa}, where the latter symmetry appears naturally in the SM for a zero Higgs mass parameter.
Einstein gravity and Higgs potentials can emerge for such models on Weyl conformal geometry in phases with spontaneously broken Weyl gauge symmetry~\cite{Ghilencea:2018dqd, Ghilencea:2019jux} and might be renormalizable~\cite{Stelle:1976gc}.
Explicit implications of Weyl gauge symmetry for our model are to be considered in future work.

\begin{table}
\begin{center}
\begin{tabular}{ |c|c|c|c| } 
 \hline
$(\#\lambda,\#\lambda')$  & $\GSM$ representation\\
  \hline\hline
  \vspace{-0.4cm} &\\
 $(1,0)$ & $\quad$ $((\mathbf{+1}\otimes \mathbf{2}\otimes \mathbf{1})\times (\mathbf{-2/3} \otimes \mathbf{1} \otimes \mathbf{3}))/\zz_6$ $\quad$  \\
  \vspace{-0.4cm} &\\
 \hline
  \vspace{-0.4cm} &\\
 $(0,1)$ & $\quad$ $((\mathbf{-1}\otimes \mathbf{\overline{2}}\otimes \mathbf{1})\times (\mathbf{+2/3}\otimes \mathbf{1} \otimes \mathbf{\overline{3}}))/\zz_6$ $\quad$ \\
  \vspace{-0.4cm} &\\
 \hline
  \vspace{-0.4cm} &\\
 $(3,0)$ & $\quad$ $\begin{array}{r}((\mathbf{-2}\otimes \mathbf{1}\otimes \mathbf{1})\times(\mathbf{-1/3}\otimes \mathbf{\overline{2}}\otimes \mathbf{\overline{3}})\qquad\\ \qquad\qquad\times\;(\mathbf{+4/3}\otimes \mathbf{1}\otimes \mathbf{3} ))/\zz_6\end{array}$ $\quad$ \\
  \vspace{-0.4cm} &\\
 \hline
  \vspace{-0.4cm} &\\
 $(0,3)$ & $\quad$ $\begin{array}{r}((\mathbf{+2}\otimes \mathbf{1}\otimes \mathbf{1})\times (\mathbf{+1/3}\otimes \mathbf{2}\otimes \mathbf{3})\qquad\\ \qquad\qquad\times \;(\mathbf{-4/3}\otimes \mathbf{1}\otimes\mathbf{\overline{3}}))/\zz_6\end{array}$ $\quad$ \\
 \hline
\end{tabular}
\end{center}
\caption{Summary of the $\GSM$ gauge transformation behavior of the fermionic, irreducible, Poincaré-irreducible $\hat{\Ocal}_{(\lambda,\lambda')}$ for column-only $\lambda, \lambda'$ with one Young diagram empty and an uneven number of boxes.
\mbox{Column 1:} Number of rows of the Young diagram pair.
\mbox{Column 2:} $\GSM$ gauge group representation, for $\Uu(1)$ denoted by the weak hypercharge, for $\SU(2)$ and $\SU(3)$ by dimension and complex conjugation.
}\label{TableTrafoBehavior}
\end{table}

The $\hat{A}_\mu(y)$ do not readily lead to well-defined Yang-Mills terms for an action functional, since the reduced gauge group $\Pp(\GL_2\rr\times \GL_3\rr)$ is non-compact~\cite{Weinberg:1996kr}.
As a projective quantum field, $\hat{A}(y)$ comes with the finite-dimensional, complex representation $\CP^4_{\pgl_5\rr}$ of $\pgl_5\rr$, which is equivalent to the representation $\CP^4_{\pgl_5\cc}$ of its complexification $\pgl_5\cc$. 
The situation is similar, if the gauge field $[\hat{A}_B([x])]$ interacts with the field operators of the projective quantum fields $\hat{\Ocal}_{(\lambda,\emptyset)}$ and $\hat{\Ocal}_{(\emptyset,\lambda)}$, which proceeds via the complex Lie algebra representations $\tilde{\rho}_{(\lambda,\emptyset)}$ and $\tilde{\rho}_{(\emptyset,\lambda)}$.
These are equivalent to the corresponding representations of $\pgl_5\cc$.

Considered for the reduced gauge group $\Pp(\GL_2\rr\times \GL_3\rr)$ with Lie algebra $\pfrak(\glfrak_2\rr\oplus \glfrak_3\rr)$, the representations are on Lie algebra level equivalent to representations of $\pfrak(\glfrak_2\cc\oplus \glfrak_3\cc)$.
Only gauge fields with values in the compact real form of $\pfrak(\glfrak_2\cc\oplus \glfrak_3\cc)$, which is
\begin{equation}
\pfrak(\ufrak(2)\oplus \ufrak(3))\cong \sfrak(\ufrak(2)\oplus \ufrak(3))\,,
\end{equation}
can provide the negative-definite kinetic terms for action functionals (positive-definite trace forms)~\cite{Weinberg:1996kr}.
The corresponding compact Lie group is $\Ss(\Uu(2)\times \Uu(3))\cong \GSM$~\cite{Baez:2009dj}, i.e., the SM gauge group.
\Cref{SecScaleTrafos} implies that the gauge bosons for a subgroup isomorphic to $\SU(2)$ do not commute with global, gauged physical scale transformations and can be massive, if this symmetry is spontaneously broken.
As for the SM, this can provide effective masses also for the fermions.

Concerning the anticipated physical equivalence of formulating the model with gauge algebra $\pfrak(\glfrak_2\rr\oplus \glfrak_3\rr)$ or $\sfrak(\ufrak(2)\oplus\ufrak(3))$, we note that they have the same real dimension and also come with equivalent finite-dimensional, complex representation theory.
Amongst others, this hints at similar particle excitation spectra.
Prominently, comparable arguments lead to the spinor representations of the complexified Lorentz algebra, which are for quantum fields on Poincaré geometry analogous to the representation $\tilde{\rho}$ for projective quantum fields.
Another example for considering the compact real form of Lie algebras is provided by the appearance of $\ufrak(N)$ gauge algebras associated with stacks of $p$-branes in string theory~\cite{Witten:1995im}, which rests upon their equivalent finite-dimensional, complex representation theory with $\glfrak_N\rr$ and charge quantization.

A summary of the behavior of the $\hat{\Ocal}_{(\lambda,\emptyset)}$ and $\hat{\Ocal}_{(\emptyset,\lambda)}$ under gauge transformations in $\GSM$ is given in \Cref{TableTrafoBehavior}.
This suggests the identification with one generation of the SM fermions, based on their behavior under gauge and Poincaré transformations, if both types of transformations are considered independently.
They are different from the SM fermions in that both transformations act via $\rho_{(\lambda,\emptyset)}$ or $\rho_{(\emptyset,\lambda)}$ and not via a tensor product of two representations.
Note also that no sterile neutrinos appear in this identification, distinct from spontaneously broken $\SU(5)$ gauge theory as a grand unified theory~\cite{Baez:2009dj}.
This is consistent with recent experiments demonstrating the absence of light sterile neutrinos~\cite{STEREO:2022nzk}.

\section{Outlook}\label{SecOutlook}

The findings of this work suggest a variety of further research questions.
First, which physical phenomena can result from interactions among the gauge bosons and the matter fields, that do not survive the non- or ultra-relativistic limits of Poincaré geometry, or that do not appear in the Poincaré limit of homogeneous, curved Lorentz geometries?

Furthermore, in how far can spontaneously, globally broken gauged physical scale symmetry provide a projective geometry-based mass generation mechanism for our model, which persists through explicit computations on the quantum level and is similar to the Higgs mechanism?
Its origin within a gauge subgroup of a non-Abelian gauge theory might hint at the absence of a corresponding Landau pole and asymptotic freedom instead~\cite{Weinberg:1996kr}.
Also, implications for potential relations among the gauge couplings of $\GSM$ gauge theory are to be investigated.

In the SM, gauge and Poincaré transformations act on the fermions via a tensor product of gauge and Poincaré group representations, in line with the Coleman-Mandula theorem, which is based on asymptotic states and Poincaré geometry~\cite{Coleman:1967ad}.
In our model, gauge and Poincaré transformations act on the fermions via $\rho_{(\lambda,\lambda')}$, gauge interactions in general take place on arbitrary geometries and asymptotic states cannot be defined on the compact $\RP^4$.
Therefore, the applicability of the Coleman-Mandula theorem appears questionable.
Yet, implications of the special relation between gauge and space-time symmetries are to be worked out for our model.
Apart from that, the particle content of one SM generation has been derived, but an explanation for the appearance of three generations and their mixing is lacking.

While our results have been derived mostly for Poincaré geometry, the framework depends smoothly on deformations of the space-time geometry.
We therefore expect that the results for Poincaré geometry hold in an approximate sense for more general space-time geometries close to Poincaré geometry.
Proximity can be defined via the structure groups with respect to available metrics on the space of closed subgroups of $\PGL_5\rr$~\cite{biringer2018metrizing}.
Probably, the results also hold approximately for weakly curved, inhomogeneous geometries. 
A suitable framework similar to the one for homogeneous geometries presented in this work and~\cite{Spitz:2024similarities} remains to be developed.

In the projective setting, geometry limits can act without singularities, from which gravity theories might benefit.
In this regard, Thomas-Whitehead gravity provides a projective description of torsion-free, metric gravity~\cite{whitehead1931representation, Eastwood2008, thomas1925projective, hall2007principle, Brensinger:2020gcv}.
It has also been noted that more general metric-affine $f(R)$ gravity comes with projectively invariant $f(R)$ action functionals~\cite{Julia:1998ys,Iosifidis:2018zjj}.
How does this relate to our model in light of the gauge bosons arising from connections of the projective frame bundle?
The analogy with a projective formulation of metric-affine gravity suggests that all particles in our model interact gravitationally.
Excitingly, we have actually shown that the gauge group for homogeneous, curved Lorentz geometries can partly reduce to the $\GL_4\rr$ of metric-affine gravity, which can also appear for Poincaré geometry.

It remains fascinating to the author that major non-trivial parts of the peculiar SM structure including hints toward the presence of gravity can be derived from projective geometry together with foundational physical concepts within the general framework of projective quantum fields.

\acknowledgments
I acknowledge fruitful exchange on the present work with S.~Flörchinger, L.~Hahn, M.~Salmhofer, S.~Schmidt and A.~Wienhard.
This work is funded by the Deutsche Forschungsgemeinschaft (DFG, German Research Foundation) under Germany’s Excellence Strategy EXC 2181/1–390900948 (the Heidelberg STRUCTURES Excellence Cluster) and the Collaborative Research Centre, Project-ID No.~273811115, SFB 1225 ISOQUANT.

\appendix

\section{Isomorphism between $\Pp(\GL_2\rr\times \GL_3\rr)$ and $\rr_{\neq 0}\times \PGL_2\rr\times \PGL_3\rr$}\label{AppendixGroupIso}
The following group isomorphisms can be established. 
The one for the special linear groups is presented for further intuition.

\begin{prop}\label{PropReducedGaugeGroupIso}
The map $\zeta:\rr_{\neq 0}\times \PGL_2\rr\times\PGL_3\rr\to \Pp(\GL_2\rr\times\GL_3\rr)$,
\begin{equation}
\zeta(a,[g],[h]) = \Pp\left(\begin{matrix}
a^3 g & \\
 & a^{-2} h
\end{matrix}\right)\,,
\end{equation}
is a group isomorphism, where $g$ is an arbitrary of the two unit-determinant representatives of $[g]\in\PGL_2\rr$ and $h$ is the unique unit-determinant representative of $[h]\in\PGL_3\rr$.
Similarly, the map $\zeta':\rr_{>0}\times\SL_2\rr\times\SL_3\rr \to \Ss(\GL_2\rr\times\GL_3\rr)\cong \Pp(\GL_2\rr\times \GL_3\rr)$,
\begin{equation}
\zeta'(a,g,h) = \left(\begin{matrix}
a^3 g & \\
 & a^{-2} h
\end{matrix}\right)\,,
\end{equation}
is a group isomorphism.
\end{prop}

\begin{proof}
Clearly, $\zeta$ is a Lie group homomorphism. 
A general element of $\Pp(\GL_2\rr\times \GL_3\rr)$ can be written as
\begin{equation}\label{EqPropReducedGaugeGroupIsoGeneralShape}
\Pp\left(\begin{matrix}
\tilde{a}^5 g  & \\
 & \tilde{b} h
\end{matrix}\right) = \Pp\left(\begin{matrix}
\tilde{a}^5\tilde{b}^{-1} g & \\
 & h
\end{matrix}\right)
\end{equation}
for $\tilde{a},\tilde{b} \in (0,\infty)$ and $g\in \SL_2\rr$, $h\in \SL_3\rr$, noting that $-1_{2\times 2}\in \SL_2\rr$. 
Setting $a := \tilde{a} \tilde{b}^{-1/5}$, the matrix \eqref{EqPropReducedGaugeGroupIsoGeneralShape} equals
\begin{equation}
\Pp\left(\begin{matrix}
a^5 g & \\
 & h
\end{matrix}\right) = \Pp\left(\begin{matrix}
a^3 g & \\
 & a^{-2} h
\end{matrix}\right)\,,
\end{equation}
which is in the image of $\zeta$; the map $\zeta$ is surjective.
In particular, this parametrization of $\Pp(\GL_2\rr \times \GL_3\rr)$ elements via $a\in (0,\infty)$, $g\in \SL_2\rr$, $h\in \SL_3\rr$ is the same as the parametrization by $a \in \rr_{\neq 0}$, $[g]\in \SL_2\rr / \zz_2$, $h\in \SL_3\rr$, and $\PGL_2\rr\cong \SL_2\rr/\zz_2$, $\PGL_3\rr\cong \SL_3\rr$.
To show injectivity, we compute the kernel and set
\begin{equation}\label{Eqetakernel}
[1] = \zeta(a,[g],[h]) = \Pp\left(\begin{matrix}
a^3 g & \\
& a^{-2} h
\end{matrix}\right)\,,
\end{equation}
for some $a\in (0,\infty)$ and $g\in\SL_2\rr$, $h\in \SL_3\rr$. 
\Cref{Eqetakernel} yields $a^3 g = a^{-2}h = b\cdot 1_{5\times 5}$ for some $b \neq 0$.
The only $g\in \SL_2\rr$ proportional to $1_{2\times 2}$ are $\pm 1_{2\times 2}$, the only $h\in\SL_3\rr$ proportional to $1_{3\times 3}$ is $1_{3\times 3}$.
Hence, $b = a^{-2}$ and $b = \pm a^3$, such that $a = \pm 1$. 
But $a\in (0,\infty)$, such that $a = 1$ and $b = 1$, which implies $g = 1_{2\times 2}$. 
The kernel of $\zeta$ is indeed trivial.

The statement for $\zeta'$ has been proven implicitly.
\end{proof}

Other than for the isomorphism $\Ss(\Uu(2)\times \Uu(3))\cong \GSM = (\Uu(1)\times\SU(2)\times\SU(3))/\zz_6$~\cite{Georgi:1974sy,Baez:2009dj}, no quotient by $\zz_6$ appears for the map $\zeta'$, since the only roots of unity in $\SL_2\rr$ are $\pm 1_{2\times 2}$, which are taken care of by restricting the Abelian subgroup to $\rr_{>0}$, and $\SL_3\rr$ has trivial center.

\section{Gauge group reduction for homogeneous, curved Lorentz geometries}\label{AppendixGaugeGrpReductionHomogLorentz}
In this appendix we derive the gauge group reduction for homogeneous, curved Lorentz geometries of the form
\begin{equation}
[g]_*\Gg(4,1) = ([g]\cdot \Xx(4,1), \Ad_{[g]} \Pp\Oo(4,1))
\end{equation}
for some $[g]\in\PGL_5\rr$.
Poincaré geometry $\Gg((1),(3,1))$ is a limit of this, see e.g.~\Cref{EqBoostLimit}.

We consider geometry deformations of $\Gg(4,1)$ by
\begin{equation}\label{EqctauGeneralLorentzian}
[c'(\kappa)] = \Ad_{[g]}[c(\kappa)]=[g]\cdot \Pp\left(\begin{matrix}
\kappa &\\
& 1_{4\times 4}
\end{matrix}\right)\cdot [g^{-1}]
\end{equation}
for $0<\kappa\leq 1$.
A straight-forward computation shows that the $\kappa\to 0$ limit of the deformed geometry $[c'(\kappa)]_*\Gg(4,1)$ is Poincaré geometry $\Gg((1),(3,1))$, see also~\cite{Spitz:2024similarities}.

Similarly to \Cref{ThmGaugeGroupReduction}, the following statement can be derived for the Poincaré limit of the geometry $[g]_*\Gg(4,1)$.
Intuitively, only gauge bosons for a gauge subgroup isomorphic to $\Pp(\rr_{\neq 0}\times \GL_4\rr) \cong \GL_4\rr$  or an Abelian gauge subgroup isomorphic to $\rr^4$ interact non-trivially with matter field operators in the Poincaré limit.

\begin{thm}\label{ThmGaugeGroupRestrictionExtensionLorentzian}
Let $\hat{\Ocal}$ be a projective quantum field, $\hat{A}$ a projective $\PGL_5\rr$ gauge quantum field and $\PGL_5\rr$ gauge transformations act on $[\hat{\Ocal}([x])]$ as in \Cref{EqActionGaugeTrafoQuantumField}. 
Consider the geometry $[g]_*\Gg(4,1)$.
Then maximally the action of gauge bosons on the projective correlators in $\Afrak_{\hat{\Ocal}}([g]\cdot \Xx(4,1))$ via covariant derivatives for the gauge subgroup $\Ad_{[g]}\Pp(\rr_{\neq 0}\times \GL_4\rr) \cong \GL_4\rr$ or an Abelian gauge subgroup isomorphic to $\rr^4$ remains non-trivial in their Poincaré limit.
\end{thm}

\begin{proof}
Consider the deformed gauge interaction operator
\begin{align}
&\sum_{C,\beta} c_{BC}'(\kappa)\Ad_{U([c'(\kappa)])} \tilde{\rho}_{\alpha\beta}([\hat{A}_C([x])]) \hat{\Ocal}_\beta([x])\nonumber\\
& =\; \sum_{\beta,\gamma}\tilde{\rho}_{\alpha\beta}([\hat{A}_B([c'(\kappa)x])]) \rho_{\beta\gamma}([c'(\kappa)^{-1}])\hat{\Ocal}_\gamma([c'(\kappa)x])\nonumber\\
& =\; \sum_{\beta,\gamma}\rho_{\alpha\beta}([c'(\kappa)^{-1}])\nonumber\\
&\quad\times \tilde{\rho}_{\beta\gamma}([c'(\kappa)\hat{A}_B([c'(\kappa)x])c'(\kappa)^{-1}]) \hat{\Ocal}_\gamma([c'(\kappa)x])\,.
\end{align}
With corresponding matrix blocks we compute
\begin{equation}\label{EqBlocksDeSitter}
\Ad_{[c'(\kappa)]}\Ad_{[g]} \Pp\left(\begin{matrix}
a & B\\
C & D
\end{matrix}\right) = \Ad_{[g]}\Pp\left(\begin{matrix}
a & \kappa B\\
C/\kappa & D
\end{matrix}\right)\,,
\end{equation}
where $a\in \rr$.
The projective matrix \eqref{EqBlocksDeSitter} has the well-defined $\kappa\to 0$ limit
\begin{equation}
\Ad_{[g]}\Pp\left(\begin{matrix}
a & 0\\
0 & D
\end{matrix}\right)\,,
\end{equation}
if $C=0$.
If $C\neq 0$, the $\kappa\to 0$ limit of \eqref{EqBlocksDeSitter} reads
\begin{equation}
\Ad_{[g]}\Pp\left(\begin{matrix}
0 & 0\\
C & 0
\end{matrix}\right)\,.
\end{equation}
Therefore, only gauge bosons for the gauge subgroup $\Ad_{[g]}\Pp(\rr_{\neq 0}\times \GL_4\rr)$ or an Abelian gauge subgroup isomorphic to $\rr^4$ can act non-trivially on the field operators of $\hat{\Ocal}$ in this limit.
\end{proof}

We call $\Ad_{[g]}\Pp(\rr_{\neq 0}\times \GL_4\rr)$ together with $\rr^4$ the reduced gauge group for homogeneous, curved Lorentz geometries.
It is not isomorphic to the reduced gauge group $\Pp(\GL_2\rr\times \GL_3\rr)$ of \Cref{ThmGaugeGroupReduction} due to the different assumptions underlying its construction.

According to \Cref{ThmGaugeGroupRestrictionExtensionLorentzian}, deformations of geometries conjugate the reduced gauge group within $\PGL_5\rr$. 
Also in limits of geometries, it remains isomorphic to $\Pp(\rr_{\neq 0}\times \GL_4\rr)$ together with $\rr^4$, as the following lemma demonstrates.

\begin{lem}\label{LemmaRedGaugeGroupLimits}
Under the assumptions of \Cref{ThmGaugeGroupRestrictionExtensionLorentzian}, the reduced gauge group remains isomorphic to $\Pp(\rr_{\neq 0}\times\GL_4\rr)$ or an Abelian subgroup of $\PGL_5\rr$ isomorphic to $\rr^4$ in limits of the geometry $\Gg(4,1)$.
\end{lem}

\begin{proof}
We employ the $KAK$ decomposition of $\PGL_5\rr$ for $K=\PO(5)$, such that $K$ is compact and $A$ is the subgroup of diagonal projective matrices in $\PGL_5\rr$, see e.g.~\cite{knapp2001representation}.
Let $[b_n]\in\PGL_5\rr$, which we can decompose as $[b_n] = [k_n a_n l_n]$, where $[k_n],[l_n]\in K$, $[a_n]\in A$.
The sequences $[k_n],[l_n]$ have subsequences $[k_{n_j}],[l_{n'_j}]$ converging to some $[k'],[l']\in K$.
Conjugation of the reduced gauge group $\Pp(\rr_{\neq 0}\times\GL_4\rr)$ with $[k']$ or $[l']$ is an isomorphism.

The adjoint action of sequences $[a_n]\in A$ on the (non-compact) reduced gauge group $\Ad_{[l']}\Pp(\rr_{\neq 0}\times\GL_4\rr)$ leaves it invariant.
Indeed, for any $[g]\in \Ad_{[l']}\Pp(\rr_{\neq 0}\times\GL_4\rr)$ we can choose a sequence $[g_n]\in \Ad_{[l']}\Pp(\rr_{\neq 0}\times\GL_4\rr)$, such that $\Ad_{[a_n]}[g_n] = [g]$ for all $n$.
Therefore, $\lim_{n\to\infty}\Ad_{[a_n l']}\Pp(\rr_{\neq 0}\times\GL_4\rr)$ is isomorphic to $\Ad_{[l']}\Pp(\rr_{\neq 0}\times\GL_4\rr)$.

For the $\rr^4$ subgroup, the argument using the possibility to choose sequences $[g_n]$ as for limits via $[a_n]$ can be repeated.
\end{proof}

Considering \Cref{LemmaRedGaugeGroupLimits}, it is to be regarded that geometry limits quite generally reduce the dimension of the space-time support of projective quantum fields and their projective correlator algebras.
In such cases, projective quantum fields restricted to limit geometries are only partially given by geometry limits of projective quantum fields on other geometries.

\section{Commutativity with physical scale transformations}\label{AppendixConformalWeight}
In \Cref{SecScaleTrafos} we have discussed global gauge transformations given by physical scale transformations.
With regard to the reduced gauge group $\Ad_{\tau}\Pp(\GL_2\rr\times \GL_3\rr)$, the following related statement is of interest.

\begin{lem}\label{LemmaWeakInteractionsBreakScaleInvariance}
Let $\zeta:\rr_{\neq 0}\times\PGL_2\rr\times \PGL_3\rr\to \Pp(\GL_2\rr\times\GL_3\rr)$ be the isomorphism of \Cref{PropReducedGaugeGroupIso}.
For $s\neq 0$ and $[1]\neq [k]\in \Ad_{\tau} \zeta(1,\PGL_2\rr,[1])$:
\begin{equation}
[d(s)^{-1} \, k^{-1}\, d(s)\, k] \neq 1\,,
\end{equation}
while this group commutator equates to one for all $[k]\in \Ad_{\tau}\Pp(1_{2\times 2}\times \GL_3)=\Ad_{\tau} \zeta(\rr_{\neq 0},[1],\PGL_3\rr)$ and all $s$.
\end{lem}

\begin{proof}
The statement for $[k]\in \Ad_{\tau}\Pp(1_{2\times 2}\times \GL_3\rr)$ is clear.
Let $[1]\neq [k]\in \Ad_{\tau}\zeta(1,\PGL_2\rr,[1])$.
As for \Cref{PropReducedGaugeGroupIso}, there is a $g\in\SL_2\rr$, such that
\begin{equation}
[k] = \Ad_{\tau} \Pp\left(\begin{matrix}
g & \\
& 1_{3\times 3}
\end{matrix}\right)\,.
\end{equation}
We write $g$ as
\begin{equation}
g=\left(\begin{matrix}
a & b\\
c & d
\end{matrix}\right)
\end{equation}
for $a,b,c,d\in \rr$, which fulfil the unit-determinant constraint $ad-bc=1$.
Thus,
\begin{equation}
[k] = \Pp\left(\begin{matrix}
a & & b\\
& 1_{3\times 3} &\\
c & & d
\end{matrix}\right)\,.
\end{equation}
Explicit computation yields
\begin{align}\label{EqGroupCommutatorExplicit}
&[d(s)^{-1} \, k^{-1}\, d(s)\, k] \nonumber\\
&\qquad = \Pp\left( \begin{matrix}
e^s ad-e^{2s} bc && (e^s -e^{2s} )bd \\
& 1_{3\times 3} &\\
(e^s -1)ac && e^s ad-bc
\end{matrix}\right)\,.
\end{align}
Assuming $[d(s)^{-1} \, k^{-1}\, d(s)\, k] = [1]$ together with the unit-determinant constraint for $g$ provides a set of five independent algebraic equations.
Certainly, $s=0$ is a solution to these.
Assume $s\neq 0$.
Then the non-trivial off-diagonal elements of the projective matrix \eqref{EqGroupCommutatorExplicit} give $bd=ac=0$.
Inserting $ad=1+bc$ (unit-determinant constraint) into $e^s ad-bc=1$ leads to $e^s +(e^s -1)bc=1$, i.e., $bc=-1$.
Therefore, $ad=0$ and with $e^s ad-e^{2s} bc=1$: $e^{2s} bc=-1$, which is a contradiction.
Only $s=0$ solves $[d(s)^{-1} \, k^{-1}\, d(s)\, k] = [1]$ for $[1]\neq [k]\in\Ad_{\tau}\zeta(1,\PGL_2\rr,[1])$.
\end{proof}

\bibliography{literature}

\end{document}